\documentclass[11pt]{article}
\usepackage{longtable,supertabular}
\usepackage{amsmath}
\usepackage{amssymb}
\usepackage{latexsym}

\usepackage{amsthm}
\usepackage{amsfonts}
\usepackage{graphicx}

\newtheorem{example} {{Example}}

\newtheorem{theorem} {{Theorem}}

\newtheorem{corollary} {{Corollary}}

\makeatletter
 \DeclareRobustCommand{\nsbinom}{\genfrac[]\z@{}}
 \makeatother
\newcommand{\tabincell}[2]{\begin{tabular}{@{}#1@{}}#2\end{tabular}}
\newcommand{\cP}{{\cal P}}
\newcommand{\cG}{{\cal G}}

\newcommand{\sbinomq}[2]{\nsbinom{{#1}}{{#2}}_{q}}
\newcommand{\deff}{\mbox{$\stackrel{\rm def}{=}$}}

\title{\bf  Construction  of Const Dimension Codes from Serval Parallel Lift MRD Code}
\author{Xianmang He, Yindong Chen}

\begin{document}

\maketitle
\begin{abstract}
In this paper, we generalize the method of using two parallel versions of the lifted MRD code  from the existing work \cite{XuChen}. The Delsarte theorem of the rank distribution of MRD codes is an important part  to count codewords in our construction.   We give a new generalize construction to the following bounds: if $n\ge k\ge d$, then $$A_q(n+k,k,d)\ge q^{n(k-\frac{d}{2}+1)}+\sum_{r=\frac{d}{2}}^{k-\frac{d}{2}} A_r(Q_q(n,k,\frac{d}{2})).$$
On this basis, we also give a construction of  constant-dimension  subspace codes from several parallel versions of lifted MRD codes.
This construction contributes to  a new lower bounds for ${A}_q((s+1)k+n,d,k)$.

{\bf keyword: constant dimension code, subspace code,  maximum rank distance code, lifted maximum rank distance code}
\end{abstract}

\section{Introduction}
Let $F_q$ be the finite field with $q>1$ elements and let $V \cong F_q^n$ be the set of all vectors of length $n$ over $F_q$. $F_q^n$ is a vector
space with dimension $n$ over $F_q$. The \emph{projective space} $\cP_q(n)$, is the set of all subspaces of $F_q^n$, including $\{ {\bf 0} \}$ and $F_q^n$. For a given integer $k$,
$1 \leq k \leq n$, let $\cG_q(n,k)$ denote the set of all
$k$-dimensional subspaces of $F_q^n$. $\cG_q(n,k)$ is often
referred to as Grassmannian. It is well known that
$$ \begin{small}
| \cG_q (n,k) | = \sbinomq{n}{k}
\deff \frac{(q^n-1)(q^{n-1}-1) \cdots
(q^{n-k+1}-1)}{(q^k-1)(q^{k-1}-1) \cdots (q-1)}
\end{small}
$$
where $\sbinomq{n}{k}$ is the $q-$\emph{ary Gaussian
coefficient}.

The set of all subspaces of $V$ forms a metric space associated with the so-called subspace distance $d_S(U,W) = \dim(U+W) - \dim(U \cap W)=2dim(U+W)-2k$.
A $(n,M,d,k)_q$ constant dimension code (CDC) $C$ is a subset of $V$ of cardinality $M$ in which for each pair of elements, the subspace distance is lower bounded by $d$, that is, for all $U \ne W \in C$, we have $d_S(U,W)\ge d$.
The main question of subspace coding in the constant dimension case asks for the maximum cardinality $M$ for fixed parameters $q$, $n$, $d$, and $k$ of a $(n,M,d,k)_q$ code.
Let $A_q(n,d,k)$ denote the maximum size of an $(n,M,d,k)_q$ code.

One way to construct good subspace codes utilizes rank-metric codes.
A linear rank metric code $[k \times n, \varrho, d]_q$ is a subspace $C$ of the vector space of $k \times n$ matrices over $F_q$, i.e., $F_q^{k \times n}$,  for which the distance of two elements is lower bounded via the rank metric $d_r(A,B) = rk(A-B)$, i.e., $d \le d_r(A,B)$ for all $A \ne B \in C$. For all parameters, $0 \le k, n, d$ and $q$ prime power, there is a linear rank metric code that attains the maximum cardinality of
$\left\lceil q^{\max\{k,n\}(\min\{k,n\}-d+1)}\right\rceil$ (see \cite{Gabidulin1985Theory}).  \\

In this paper, we extend the method of using two parallel versions of lifted MRD codes to a  general framework \cite{XuChen}. This method  gives us a new lower bound,
and leads to several new lower bounds from several parallel versions of lifted MRD codes.

\section{Preliminaries}

\subsection{Lifted MRD code}
 Without loss of generality, we assume that $n\ge k$.  For any given MRD code $Q_q(n,k,d)$ with the rank distance $d$, we have a $(n+k, q^{n(k-d+1)}, 2d,n)_q$ CDC consisting of $q^{n(k-d+1)}$ subspaces of dimension $k$  in ${\bf F}_q^{n+k}$  is defined as  $lift(U_A)=\{rowspace[I_k, A]\}$, spanned by rows of $(I_k,A)$. Here $I_k$ is the $k\times k$ identity matrix. For any  $A$ and $B$, the subspaces $U_A$ and $U_B$ spanned by rows of $(I_k,A)$ and $(I_k,B)$ are the same if and only if $A=B$. The intersection $U_A \cap U_B$ is the set $\{ \alpha A: \alpha A=\alpha B, \alpha \in {\bf F}_q^k\}$. Thus $dim(U_A \cap U_B) \leq k-rank(A-B) \leq k-d$. The distance of this CDC is $2d$. A CDC  constructed as above is called a lifted  MRD  code $C^{MRD}$.\\

\subsection{Delsarte Theorem}

For $r \in {\bf Z}^{+}$, the rank distribution of a code $Q$ in $Q_q(m,n,d)$ is defined by $A_r(Q)=|\{Q \in {Q_q(m,n,d)}, rk(Q)=r\}|$  (see \cite{Delsarte1978Bilinear,Cruz2015Rank}).  We refer the following result to Theorem 5.6 in \cite{Delsarte1978Bilinear} or Corollary 26 in \cite{Cruz2015Rank}.  \\

\begin{theorem}\label{them-delsarte}
{\bf(Delsarte 1978)} {\em Assume that $Q \subseteq {Q_q(m,n,d)} $ ($m\ge n$)is a MRD code with rank distance $d$, then its rank distribution is given by $$A_r(Q)=\displaystyle{n \choose r}_q \Sigma_{i=0}^{r-d} (-1)^i q^{\displaystyle{i \choose 2}} \displaystyle{r \choose i}_q (\frac{q^{m(n-d+1)}}{q^{m(n+i-r)}}-1).$$ }
\end{theorem}

Here, the value of $r$  varies from  $d$ to  $min\{m,n\}(A_0(Q)=1)$, and the rank distribution of a MRD code is determined by its parameters $m,n,d,q$.

\begin{example}\label{exam-delsate}
 Assume that $m=n=4, q=2, d=2$, we have $|Q_2(4,4,2)|=2^{12}$, and $A_2(Q_2(4,4,2))=525, A_3(Q(4,4,2))=2250,A_4(Q_2(4,4,2))=1320$. Similarly, $m=5, n=5, d=2, q=2$, then  $A_2(Q_2(5,5,2))=4805, A_3(Q_2(5,5,2)))=124930$. Moreover, $A_2(Q_2(6,5,2))=9765$, $A_3(Q_2(6,5,2))=566370$,$A_2(Q_2(7,5,2))=19685,A_3(Q_2(7,5,2))=2401570$.
\end{example}

\subsection{Previous constructions}

There is a variety of lower bounds and upper bounds for constant dimension codes.
 For numerical
values of the known lower and upper bounds on the sizes of subspace codes we refer
the reader to the online tables at http://subspacecodes.uni-bayreuth.de associated
with \cite{Heinlein2016Tables}.

Koetter and Kschischang~\cite{koetter2008coding}, Etzion and Vardy~\cite{tuviztzion2011analogs} developed several upper bounds on $A_q(n,d,k)$.
Surveys of upper bounds for constant dimension codes can be in the paper \cite{Khaleghi2009Subspace}.
One of the Johnson type bound in \cite{Etzion2011Error} (see Theorem 4 in \cite{Etzion2011Error}) of CDC is
\begin{equation}
A_q(n, 2\delta, k) \leq \frac{\displaystyle{n \choose k-\delta+1}_q}{\displaystyle{k \choose k-\delta+1}_q}.
\end{equation}

The iterative application of the Johnson type bound II (\cite[Theorem~3]{Xia2009Johnson},~\cite[Theorem~4,5]{Etzion2011Error}), which is a $q$-generalization of \cite{johnson1962new}, gives the upper bound
\begin{equation}
\label{ie_r_johnson}
A_q(n,d,k) \le
\left\lfloor \frac{q^{n}-1}{q^{k}-1} \left\lfloor \frac{q^{n-1}-1}{q^{k-1}-1} \left\lfloor \ldots
\left\lfloor \frac{q^{n'+1}-1}{q^{\frac{d}{2}+1}-1} A_q(n',d;\frac{d}{2}) \right\rfloor
\ldots \right\rfloor \right\rfloor \right\rfloor
\end{equation}
where $n' = n - k + \frac{d}{2}$.  It is attained with equality at $n=ak$ and   $d=2k$, i.e., and also at $n=13$, $k=3$,   $d=4$ with   $A_2(13,4;3)=1597245$, see~\cite{Braun2016Existence}.

While a lot of upper bounds for the maximum sizes of CDCs have been proposed in the literature, inequalities dominates most of them. \cite{Heinlein2017Classifying} determined the maximum size $A_2(8,6,4)$ of a binary subspace code to be 257. Another known improvement is $A_2(6,4,3) = 77.$
Some prior results about general lower bounds for $A_q(n,d,k)$ can see the references  \cite{silberstein2014subspace,Silberstein2013,Gluesing2017linkage,heinlein2017asymptotic}. Many CDC's from the multilevel construction based on echelon-Ferrers diagram have been given\cite{etzion2009error}.
Authors of \cite{TrautmannRosenthal2010} improved the echelon-Ferrers construction by a technology termed as  the pending dot.
The pending dot was extended to pending block in the paper \cite{silberstein2014subspace}.

The improved linkage construction is one of the most powerful methods in network coding, and it produce currently the best known lower bound about 69.1\% of the constant dimension code\cite{Gluesing2017linkage}.
The work that is closest to ours is \cite{XuChen}, which gives a construction of constant-dimension subspace
codes from two parallel versions of lifted MRD codes.
If $2t \ge n$, then   $A_q(2n, 2(n-t), n) \ge q^{n(t+1)}+\sum_{r=n-t}^n A_r(Q_q(n,n,t))$.

In this paper, we extend the above construction to a  general case $A_q(n+k,d,k), n\ge k \ge d$.

\section{Our construction}

The basic idea is as follows:  Define the set $W_1=\{rowspace[A,I_k]| A \in Q_q(n,k,\frac{d}{2})\}$ be a  $(n+k, q^{n(k-\frac{d}{2}+1)},d,k)_q$ CDC  in $F_q^{n+k}$ spanned by the rows of $k \times (n+k)$ matrices $(A,I_k)$ with $q^{n(k-\frac{d}{2}+1)}$ elements. Similar to code $W_1$, we have  another $k$ dimension subspaces in ${\bf F}_q^{n(k-\frac{d}{2}+1)}$ spanned by the rows of $W_2=(B,I_k)$, where $B$ is in the MRD code $Q_q(n,k,\frac{d}{2})$, can be used to increase the size of the constructed constant-dimension subspace codes.
Now the problem is how many different subspaces we can take from these two parallel versions of lifted MRD code such that
the subspace distance $d$ is preserved.  The key point here is to keep the subspace distances larger than or equal to $d$ by suitable conditions. By using this idea  a new lower bound for  $A_q(n+k,d, k), k\geq d$ are given with the help from the Delsarte Theorem about the rank distributions of the MRD code $Q_q(n,k,d)$.

\subsection{A New Lower Bound}
Our construction is

\begin{theorem} If $n\ge k\ge d$,  then
$$A_q(n+k,k,d)\ge q^{n(k-\frac{d}{2}+1)}+\sum_{r=\frac{d}{2}}^{k-\frac{d}{2}} A_r(Q_q(n,k,\frac{d}{2}))$$ \label{them-new-main-code}
\end{theorem}

\begin{proof}
 Let $Q_1 \subset Q_q(n,k,\frac{d}{2})$ be a MRD code with rank distance $\frac{d}{2}$, and  let $Q_2 \subset Q_q(n,k,\frac{d}{2})$ be a MRD code with the rank of each element in $Q_2$ is at most $k-\frac{d}{2}$.  The code is defined by $${\bf C}=\{(I_k|Q_{12}):  Q_{12} \in Q_1\} \cup \{(Q_{21}|I_k): Q_{21} \in  \bf{Q_2}\}.$$ From the definition of lift MRD code, the subspace distances of the two codes $${\bf W}_1=\{(I_k|Q_{12}): Q_{12} \in Q_1\}$$ and $${\bf W}_2=\{(Q_{21}|I_k): Q_{22} \in Q_2\}$$ are at least $d$. We only need to prove that the subspace distance of $W_1 \in {\bf W}_1$ and $W_2 \in {\bf W}_2$ is at least $d$. It is clear that  these two codes are disjoint.\\

It is sufficient to prove that
$$dim(W_1+W_2)=  rank\left(
\begin{array}{cccc}
I_k &  Q_{12} \\
Q_{21}& I_k
\end{array} \right) \ge k+\frac{d}{2}$$

We split the $Q_{12}$ into two parts: $Q_{121}$ and $Q_{122}$, where $Q_{121}$ is a matrix with $(n-k)\times k$  and  $Q_{122}$ is a matrix with $k\times k$.
$$dim(W_1+W_2)=  rank\left(
\begin{array}{cccc}
I_k & Q_{121} & Q_{122}\\
Q_{21} &  &I_k
\end{array} \right) \ge k+\frac{d}{2}$$

The above formula can be transformed into the following  by subtracting second row multiplied by $Q_{122}$:

$$dim(W_1+W_2)=  rank\left(
\begin{array}{cccc}
I_k | Q_{121}- Q_{21}\times Q_{122}  & 0\\
Q_{21} &  &I_k
\end{array} \right) $$

It is clear that $dim(W_1+W_2)\ge k+rank(I_k | Q_{121}-Q_{21}\times Q_{122})$. We note that the rank of  $Q_{21}$ is at most $k-\frac{d}{2}$, then  $rank(I_k | Q_{121} -Q_{21}\times Q_{122})\ge k-(k-\frac{d}{2})\ge \frac{d}{2}$.
In the end, we have $dim(W_1+W_2)\ge k+\frac{d}{2}$. Here $(I_k|Q_{121})$ is a $k \times n$ matrix concatenated from $I_k$ and $Q_{121}$.

\end{proof}

\section{Parallel Construction from MRD Codes}

\subsection{General Construction}

Our construction is as follow. For simplify, the subset of a MRD code $Q_q(n,k,\frac{d}{2})$ with the rank at most $k-\frac{d}{2}$ is denoted by $SQ_q(n,k,\frac{d}{2})$.

\begin{theorem}\label{them-several-main}
If $n \ge k \ge d, s\ge 0$, then $A_q((s+1)\times k+n, d, k) \ge$  \\ $$\sum_{j=0}^{s}q^{((s-j)k+n)(k-\frac{d}{2}+1)}(\sum_{r=\frac{d}{2}}^{k-\frac{d}{2}}A_r(Q_q(k,k,\frac{d}{2})))^j +(\sum_{r=\frac{d}{2}}^{k-\frac{d}{2}}A_r(Q_q(n,k,\frac{d}{2})))\times  (\sum_{r=\frac{d}{2}}^{k-\frac{d}{2}} A_r(Q_q(k,k,\frac{d}{2})))^{s-1}.$$
\end{theorem}

\begin{proof}  For the first block $B_1$, we take $k$ dimension subspaces in $F_q^{(s+1)k+n}$ spanned by rows of $(I_k, A_1^1,\ldots, A_s^1)$ where $A_1^1,\ldots,A_{s-1}^1$ are from the MRD code $Q_q(k,k,\frac{d}{2})$, $A_s^1$ takes from the MRD code $Q_q(n,k,\frac{d}{2})$. There are $q^{(sk+n)(k-\frac{d}{2}+1)}$ such subspaces.

For the second block $B_2$, we take $k$ dimension subspaces in $F_q^{(s+1)k+n}$ spanned by rows of $(A_1^2,I_k,\ldots, A_s^2)$ where $A_1^2$ takes from $SQ_q(k,k,\frac{d}{2})$. $A_2^2\ldots,A_{s-1}^2$ are from the MRD code $Q_q(k,k,\frac{d}{2})$ and $A_s^2 \in Q_q(n,k,\frac{d}{2})$. There are $q^{((s-1)k+n)}\times |SQ_q(n,k,\frac{d}{2})|=q^{((s-1)k+n)}\times (\sum_{r=\frac{d}{2}}^{k-\frac{d}{2}} A_r(Q_q(k,k,\frac{d}{2}))) $ such subspaces.

For the third block, we take $k$ dimension subspaces in ${F}_q^{(s+1)k+n}$ spanned by rows of $( A_1^3,A_2^3, I_k,\ldots, A_s^3)$ where $A_1^3,A_2^3$ takes from $SQ_q(k,k,\frac{d}{2})$.  $A_2^3\ldots,A_{s-1}^3$ are from the MRD code $Q_q(k,k,\frac{d}{2})$, and $A_s^3$ takes from the MRD code $Q_q(n,k,\frac{d}{2})$. There are $q^{((s-2)k+n)(k-\frac{d}{2})}\times (\sum_{i=\frac{d}{2}}^{k-\frac{d}{2}+1} A_i(Q_q(k,k,\frac{d}{2})))^2$ such subspaces.
We can continue this process. It is obvious that all these subspaces in ${\bf F}_q^{(s+1)k+n}$ are different.

These blocks can be demonstrated  as follows:

$$\left(
\begin{array}{lllllll}
B_1=(I_k& Q_q(k,k,\frac{d}{2})& Q_q(k,k,\frac{d}{2})& \cdots & Q_q(n,k,\frac{d}{2}))\\
B_2=(SQ_q(k,k,\frac{d}{2})&I_k& Q_q(k,k,\frac{d}{2})& \cdots & Q_q(n,k,\frac{d}{2}))\\
B_3=(SQ_q(k,k,\frac{d}{2})& SQ_q(k,k,\frac{d}{2})& I_k& \cdots & Q_q(n,k,\frac{d}{2}))\\
\cdots & \cdots &\cdots & \cdots &\cdots \\
B_{s+2}=(SQ_q(k,k,\frac{d}{2})& SQ_q(k,k,\frac{d}{2})& \cdots & SQ_q(n,k,\frac{d}{2})& \quad \quad \quad I_k)\\
\end{array} \right) $$

where all the lower triangles are taken from $SQ_q(k,k,\frac{d}{2})$ or $SQ_q(n,k,\frac{d}{2})$.

For any one fixed block position $j$ of $I_k$, the dimension of the intersection of two different subspaces is at most $t$ since $A_1^j,\ldots,A_s^j$ are in the MRD code $Q_q(k,k,\frac{d}{2})$ or $Q_q(n,k,\frac{d}{2})$. For different block positions $U_j, U_i, j>i$ in the set $\{1,\ldots,s+2\}$ of $I_k$, the dimension of the sum of two different subspaces is at least $k+\frac{d}{2}$ from theorem \ref{them-new-main-code}. We get the conclusion.\\
\end{proof}

\subsection{Examples}

  When $s=0$, this is the case, which happens to the theorem \ref{them-new-main-code}.  When $s=1$, in this case we refer to table \ref{tab:A-q-5-4} and the following result can be proved.

\begin{corollary}
If $n\ge k \ge d$,  we have $A_q(2k+n,d,k) \geq q^{(k+n)(k-\frac{d}{2}+1)}+q^{n(k-\frac{d}{2}+1)}\times (\sum_{r=\frac{d}{2}}^{k-\frac{d}{2}}A_r(Q_q(k,k,\frac{d}{2}))+(\sum_{i=\frac{d}{2}}^{k-\frac{d}{2}} A_r(Q_q(k,k,\frac{d}{2}))) \times (\sum_{i=\frac{d}{2}}^{k-\frac{d}{2}}A_r(Q_q(n,k,\frac{d}{2})) $.

\end{corollary}

For example, when $n=5, k=5, d=4$, then $A_q(15,4,5)\ge q^{40}+(A_2(Q_q(5,5,2))+A_3(Q_q(5,5,2))\times q^{20}+(A_2(Q_q(5,5,2))+A_3(Q_q(5,5,2)))^2$.
Assume that $q=2$, then $A_2(15,4,5)\ge 2^{40}+(4805+124930)\times 2^{20}+(4805+124930)^2=1252379805361$,  exceeds the current known code $A_2(15,4,5)\ge 1235787711790$ which was constructed by the linkage.

When $n=6, k=5, d=4$, then $A_q(16,4,5)\ge q^{44}+(A_2(Q_q(5,5,2))+A_3(Q_q(5,5,2))\times q^{24}+(A_2(Q_q(5,5,2))+A_3(Q_q(5,5,2)))\times (A_2(Q_q(6,5,2))+A_3(Q_q(6,5,2)))$.
Assume that $q=2$, then $A_2(16,4,5)\ge 19843523036401$, while the current known code is 19772603404689.

When $n=7, k=5, d=4$, then $A_q(17,4,5)\ge q^{48}+(A_2(Q_q(5,5,2))+A_3(Q_q(5,5,2))\times q^{28}+(A_2(Q_q(5,5,2))+A_3(Q_q(5,5,2)))\times (A_2(Q_q(7,5,2))+A_3(Q_q(7,5,2)))$.
Assume that $q=2$, then $A_2(17,4,5)\ge 316614572112241$, while the current known code is 316361655057323.

When $n=8, k=5, d=4$, then $A_q(18,4,5)\ge q^{52}+(A_2(Q_q(5,5,2))+A_3(Q_q(5,5,2))\times q^{32}+(A_2(Q_q(5,5,2))+A_3(Q_q(5,5,2)))\times (A_2(Q_q(8,5,2))+A_3(Q_q(8,5,2)))$.
Assume that $q=2$, then $A_2(18,4,5)\ge 5062094281261681$, while the current known code is 5061786480788587.

When $n=6, k=6, d=4$, then $A_q(18,4,6)\ge q^{60}+(\sum_{r=2}^4 A_r(Q_q(6,6,2)))\times q^{30}+(\sum_{r=2}^4 A_r(Q_q(6,6,2)))^2$.
Assume that $q=2$, then $A_2(18,4,6)\ge 1321055665352277121$, while the current known code is 1301902384896972957.

When $n=7, k=6, d=4$, then $A_q(19,4,6)\ge q^{65}+(\sum_{r=2}^4 A_r(Q_q(6,6,2)))\times q^{35}+(\sum_{r=2}^4 A_r(Q_q(6,6,2)))\times (\sum_{r=2}^4 A_r(Q_q(7,6,2)))$.
Assume that $q=2$, then $A_2(19,4,6)\ge 41829335877977260673$, while the current known code is 41660876316712223851.

When $n=6, k=6, d=6$, then $A_q(18,6,6)\ge q^{48}+(A_3(Q_q(6,6,3)))\times q^{24}+(A_3(Q_q(6,6,3)))^2$.
Assume that $q=2$, then $A_2(18,6,6)\ge$ 1321055665352277121, while the current known code is 1301902384896972957.

When $n=7, k=6, d=6$, then $A_q(19,6,6)\ge q^{52}+(A_3(Q_q(6,6,3)))\times q^{28}+(A_3(Q_q(6,6,3)))\times (A_3(Q_q(7,6,3)))$.
Assume that $q=2$, then $A_2(19,6,6)\ge$ 41829335877977260673, while the current known code is 41660876316712223851.

\begin{longtable}{|l@{\extracolsep{\fill}}|l|l|l||}
 \caption{\label{tab:A-q-5-4}New subspace codes from parallel linkage}\\
\hline
$A_q(n,k,d)$　&　New 　&　Old　 \\ \hline  \hline \endfirsthead
\multicolumn{4}{r}{continued table} \\ \hline
$A_q(n,k,d)$ & New & Old \\  \endhead
$A_2(15,4,5)$ & 1252379805361   &1235787711790     \\ \hline
$A_3(15,4,5)$ & 12399152568347096641   &12394544365887696067   \\ \hline
$A_4(15,4,5)$ &\tabincell{l}{1215514411238392851780481}   &     \tabincell{l}{1215478900794081741379237}  \\ \hline
$A_5(15,4,5)$ &\tabincell{l}{9113715532351043940956916001}   &     \tabincell{l}{9113676963739967346201192181}   \\ \hline
$A_7(15,4,5)$ &\tabincell{l}{6369953433032789460601458266\\169601}   &        \tabincell{l}{6369951878418978850938882154\\998943}  \\ \hline
$A_8(15,4,5)$ &\tabincell{l}{1329603936275508669606118276\\013276161}   &     \tabincell{l}{1329603830010446369320349184\\800629897}  \\ \hline
$A_9(15,4,5)$ &\tabincell{l}{1478344516592412787455586580\\29146634561}   &   \tabincell{l}{1478344472192502033634129606\\95716746417}   \\ \hline

$A_2(16,4,5)$ & 19843523036401   &19772603404689     \\ \hline
$A_3(16,4,5)$ & 1004000504591772043921   &1003958093636913086356   \\ \hline
$A_4(16,4,5)$ &\tabincell{l}{311163172623815098594853761}   &     \tabincell{l}{311162598603284926601722789}  \\ \hline
$A_5(16,4,5)$ &\tabincell{l}{569604907196437900170969666\\6001}   &             \tabincell{l}{569604810233747959140019927\\4056}   \\ \hline
$A_7(16,4,5)$ &\tabincell{l}{152942545364298885085407927\\47953212001}   &     \tabincell{l}{152942544600839682211042601\\25199891012}  \\ \hline
$A_8(16,4,5)$ &\tabincell{l}{544605729453422094767666867\\6810298670081}   &    \tabincell{l}{544605728772278832873615029\\1737261978761}  \\ \hline
$A_9(16,4,5)$ &\tabincell{l}{969941808565538232088309356\\467880730794721}   &  \tabincell{l}{969941808205500584267352435\\307639908204424}   \\ \hline

$A_2(17,4,5)$ & 316614572112241   &316361655057323     \\ \hline
$A_3(17,4,5)$ & 81320951591684518802401   &81320605584592333256896   \\ \hline
$A_4(17,4,5)$ &\tabincell{l}{796576337880209517283482182\\41}   &                    \tabincell{l}{796576252424409420394019074\\93}  \\ \hline
$A_5(17,4,5)$ &\tabincell{l}{356003008694948362952505652\\0416001}   &               \tabincell{l}{356003006396092474470158189\\5765556}   \\ \hline
$A_7(17,4,5)$ &\tabincell{l}{367215049622878118441926070\\65053113534401}   &       \tabincell{l}{367215049586616076988713969\\88494253570488}  \\ \hline
$A_8(17,4,5)$ &\tabincell{l}{223070506509377829926955265\\58406437716147201}   &     \tabincell{l}{223070506505125409945032726\\04052070322817161}  \\ \hline
$A_9(17,4,5)$ &\tabincell{l}{636378820366486886124356752\\3072271636040343681}   &   \tabincell{l}{636378820363628933337809933\\8862136580032063628}   \\ \hline

$A_2(18,4,5)$ & 5062094281261681   &5061786480788587     \\ \hline
$A_3(18,4,5)$ & 6586968939449073017068081   &6586969052351977742082856   \\ \hline
$A_4(18,4,5)$ &\tabincell{l}{203923520267181222664011458\\32321}   &                  \tabincell{l}{203923520620648811617657149\\36261}  \\ \hline
$A_5(18,4,5)$ &\tabincell{l}{222501878974454212088292624\\3764166001}   &             \tabincell{l}{222501878997557796543846638\\9564044756}   \\ \hline
$A_7(18,4,5)$ &\tabincell{l}{881683334056450410727244356\\60598758561216801}   &     \tabincell{l}{881683334057465200849902241\\56399442519707072}  \\ \hline
$A_8(18,4,5)$ &\tabincell{l}{913696794644823698671485722\\17565157707505515521}&    \tabincell{l}{913696794644993679134854045\\86035285169051407881}  \\ \hline
$A_9(18,4,5)$ &\tabincell{l}{417528144040561519972171750\\13157922612937898347041}   &\tabincell{l}{417528144040576943162937097\\62272976069664226251756}   \\ \hline

$A_2(18,4,6)$ &  1321055665352277121 & 1301902384896972957     \\
\hline
$A_3(18,4,6)$   &  43241984454039791949376848001 &  43225562953761729683056546744     \\
\hline
$A_4(18,4,6)$  & \tabincell{l}{13364977346615645679038498701\\19608321} &                       \tabincell{l}{13364584050324721907495003191\\15666769}    \\
\hline
$A_5(18,4,6)$  & \tabincell{l}{86915431345555286301049529274\\6726010500001} &                  \tabincell{l}{86915062398553386450111346455\\8715816063570 }    \\
\hline
$A_7(18,4,6)$  & \tabincell{l}{50827312139771315191417379894\\7508628845999547723521} &         \tabincell{l}{50827299725042503817812207998\\9337055565420133852250}    \\
\hline
$A_8(18,4,6)$  & \tabincell{l}{15329290735337201203431548481\\57539946320174365857546241} &     \tabincell{l}{15329289509595962385976540015\\68049806785911717931336256}    \\
\hline
$A_9(18,4,6)$  & \tabincell{l}{17973218529883895304078740000\\31880315113074804045244546241} &  \tabincell{l}{17973217989922197607083648785\\65965820546280222286535998208}    \\
\hline

$A_2(19,4,6)$ & 41829335877977260673  & 41660876316712223851    \\ \hline
$A_3(19,4,6)$ & 10504270152299418377046931486963  &                            10503811797764100313173626438410  \\ \hline
$A_4(19,4,6)$ &\tabincell{l}{13685359611640410968142696785811\\89910529}   &           \tabincell{l}{1368533406753251523327488538756\\265820613}  \\ \hline
$A_5(19,4,6)$ &\tabincell{l}{271609616470342511721762551559160\\7311080812501}   &   \tabincell{l}{2716095699954793272557435557305\\913288978107256}   \\ \hline
$A_7(19,4,6)$ &\tabincell{l}{854254430749233338554217116630539\\7881306156458444082343}   &  \tabincell{l}{8542544264787894328644239651424\\668612867543534298440406}  \\ \hline
$A_8(19,4,6)$ &\tabincell{l}{502310159279228668519841248394562\\60102624467996257072709633}  &   \tabincell{l}{5023101586504404954636792632338\\1856072893849422409772176905}  \\ \hline
$A_9(19,4,6)$ &\tabincell{l}{106130054948108187145939004193505\\490502763166947652518585333609}&     \tabincell{l}{1061300549086915846347321188422\\58778340200008478852505231237828\\}   \\ \hline

$A_2(18,6,6)$  & 282957166112041   &   282206169223861  \\
\hline
$A_3(18,6,6)$  & 79773409708059646924801 &   7977052899429695519499   \\
\hline
$A_4(18,6,6)$  &  79228596837171602219181433561 & 79228465213535437618551984193    \\
\hline
$A_5(18,6,6)$ & \tabincell{l}{355271606149055831666451347994\\5761} &              \tabincell{l}{355271549860537797212597654883\\4375} \\
\hline
$A_7(18,6,6)$ & \tabincell{l}{367033693031672327723398953389\\21195414401} &       \tabincell{l}{367033691269048237620480816438\\30813838569}    \\
\hline
$A_8(18,6,6)$ & \tabincell{l}{223007453917574044765606725592\\19358376203601}&     \tabincell{l}{223007453646901902254328287720\\81255730905601} \\
\hline
$A_9(18,6,6)$ & \tabincell{l}{636268545986545104493692756885\\8327487086310721} &  \tabincell{l}{636268545755947049961803989258\\2186036406787787}    \\
\hline

$A_2(19,6,6)$ & 4527206647567081    & 4515298730748862   \\ \hline
$A_3(19,6,6)$ & 6461646138903634303206481    &  6461417369472937542117973  \\ \hline
$A_4(19,6,6)$ & \tabincell{l}{202825207897159451521448163\\21241}   &               \tabincell{l}{20282487415579548140041494\\597697}  \\ \hline
$A_5(19,6,6)$ & \tabincell{l}{222044753843060819393754856\\8627695761}   &          \tabincell{l}{22204471885174522176980972\\29003922001}   \\ \hline
$A_7(19,6,6)$ & \tabincell{l}{881247896969044489775331238\\06250729756056801}   &   \tabincell{l}{88124789274625593704300569\\118173789808207533}  \\ \hline
$A_8(19,6,6)$ & \tabincell{l}{913438531246383218917381666\\27038002732908483921}  & \tabincell{l}{91343853013939625003475366\\792854319199699599873}  \\ \hline
$A_9(19,6,6)$ & \tabincell{l}{417455793021772239445991928\\87121943049164300832481}&\tabincell{l}{41745579287064298367037383\\503578317354686374220297}   \\ \hline
\end{longtable}

\section{Conclusion}

In the paper  a parallel construction from MRD codes are given, which is adopted from the existing work \cite{XuChen}. A new lower bounds on $A_q(n,d,k)$ can be proved.
In essence, our result is a generalization result of the paper \cite{XuChen} when $n$ is equal to $k$.
In addition, this  method is generalized to several parallel versions of maximum rank distance codes, and this method outperform the linkage construction in some cases.
These new codes are listed in Table \ref{tab:A-q-5-4}.

\bibliographystyle{IEEEtran}
\bibliography{subspacecode}

\end{document}